\documentclass[conference]{IEEEtran}
\usepackage{graphics,graphicx,xcolor}
\usepackage{amsmath,amssymb,epsfig,dsfont}
\usepackage{verbatim,subfigure}
\usepackage{url}
\IEEEoverridecommandlockouts
\title{Round-Robin Streaming with Generations}

\author{\IEEEauthorblockN
{
Yao Li\IEEEauthorrefmark{1},
P\'{e}ter Vingelmann\IEEEauthorrefmark{2}\IEEEauthorrefmark{4},
Morten Videb{\ae}k Pedersen\IEEEauthorrefmark{4},
Emina Soljanin\IEEEauthorrefmark{3}
}\\
\IEEEauthorblockA
{
	\IEEEauthorrefmark{1}WINLAB, ECE Department, Rutgers University\\
	\IEEEauthorrefmark{2}Budapest University of Technology and Economics, Hungary\\
	\IEEEauthorrefmark{4}Aalborg University, Denmark\\	
	\IEEEauthorrefmark{3}Bell Labs, Alcatel-Lucent
}
}

\begin{document}
\maketitle
\begin{abstract}
We consider three types of application layer coding for streaming over lossy links: random linear coding, systematic
random linear coding, and structured coding. The file being streamed is divided into sub-blocks (generations). Code symbols are formed by combining data belonging to the same generation, and transmitted in a round-robin fashion. We compare the schemes based on delivery packet count, net throughput, and energy consumption for a range of generation sizes. We determine these performance measures both analytically and in an experimental configuration. We find our analytical predictions to match the experimental results. We show that coding at the application layer brings about a significant increase in net data throughput, and thereby reduction in energy consumption due to reduced communication time. On the other hand, on devices with constrained computing resources, heavy coding operations cause packet drops in higher layers and negatively affect the net throughput. We find from our experimental results that low-rate MDS codes are best for small generation sizes,  whereas systematic random linear coding has the best net throughput and lowest energy consumption for larger generation sizes due to its low decoding complexity.
\end{abstract}

\newtheorem{theorem}{Theorem}
\newtheorem{lemma}[theorem]{Lemma}
\newtheorem{claim}[theorem]{Claim}
\newtheorem{corollary}[theorem]{Corollary}
\newtheorem{remark}{Remark}

\section{Introduction}
\label{sec:intro}

With the rapid increase in multicast streaming applications, we see more and more proposals for application-layer rateless erasure coding.
A number of these schemes have already been standardized and are currently being considered for implementation,
such as Raptor codes \cite{RaptorMonograph} for Multimedia Broadcast/Multicast Service (MBMS) \cite{3gppMBMS}. The goal of these schemes is to combat transport-layer packet losses.

In packet-based data networks, large files are usually segmented into smaller blocks that are put into transport packets.
Packet losses occur not only because of the physical channel limitations between the sender and the receiver, but also when the sender pushes data at rates that exceed the speed at which the receiver can take in packets, given its limited processing power and buffer space. In point-to-point scenarios, the sender can adjust its transmission rate to avoid packet losses, and retransmit lost packets according to the feedback from the receiver, to ensure efficient and reliable data delivery. Thus unicast applications usually implement some ARQ protocol. In broadcast applications from a single sender to many receivers, however, it is costly for the sender to collect and respond to individual receiver feedbacks, and thus packet losses are inevitable.

We here consider three rateless coding schemes to combat random packet losses in single-hop scenarios. Two are based on random linear coding,
and the third is based on structured MDS coding such as Reed-Solomon (RS). Note that our scenario of interest is wireless streaming, rather than transmission over networks as in random linear network coding \cite{Ho06arandom} over generations \cite{chou2003pnc}. All schemes follow the round-robin scheduling. Since there is no feedback until the entire file has been downloaded, the round-robin protocol may result in many superfluous
transmissions for already decoded generations. We study the schemes both theoretically and by experiment.

This paper is organized as follows:
In Section~\ref{sec:model}, we introduce our coding and scheduling models and define our performance measures.
In Section~\ref{sec:theory}, we present an analytical analysis of the schemes.
In Section~\ref{sec:results}, we describe the experimental setup, and present measurement results collected on a mobile platform.
In Section~\ref{sec:concl}, we discuss the results and future work. 
\section{System Model}
\label{sec:model}

We consider transmission without feedback over a memoryless binary erasure channel between the sender and the receivers.
In a packet network, the erasure rate is evaluated as the packet loss rate, denoted as $\epsilon.$ For the theoretical analysis,
we assume that $\epsilon$ stays constant, regardless of time and the transmission protocol.

\subsection{Performance Measures}
We will measure the performance of the system by the {\it delivery packet count, delivery time, and energy consumption}.
Delivery packet count is defined as the number of packets that have to be sent until the receiver is able to recover the entire file.
Delivery time is the time the receiver has to spend in the system until it is able to recover the content. It is a random variable that depends on the delivery packet count, the packet size, and the rate of data transmission.
In a wireless network, energy consumption mainly depends on the delivery time and the transmission power, since the power consumption in transmission is dominating.

\subsection{Coding within Generations}
Suppose a file is segmented into $N$ blocks for transmission. A block fits into the data payload of an application layer packet. Throughout the
paper, we use the words ``block'' and ``packet'' interchangeably. To combat random packet losses, we apply erasure codes at the block level.
That is, instead of transmitting the original file blocks, the sender transmits a coded block formed from the original blocks. The coded block contains the same number of information bits contained in an original file block.
 But due to practical concerns, such as computational complexity, when $N$ is large, the erasure codes are
 applied to subsets of the $N$ blocks. These subsets are referred to as generations, and each file block belongs to at least one of the generations. Suppose there are $n$ generations, denoted by $G_1,G_2,\dots,G_n,$ and assume a uniform generation size of $g.$ Note that when $g=1,$ the coded blocks are effectively the original blocks. In each transmission, the sender selects one of the $n$ generations, and sends a coded block composed from the selected generation. A transmission scheme is therefore defined by three aspects: the composition of generations, the encoding scheme of blocks within each generation, and the order of selecting generations whence a coded block is created. The last component is referred to as generation scheduling. When the generation size $g=1,$ it is simply the question of which block to send in each transmission. At the receiver, coded blocks are classified by their originating generations, and decoding is performed within each generation. 

In \cite{generation_ITtrans}, the delivery time of coding within both disjoint and overlapping generations has been studied when generations are scheduled at random and when coded blocks are random linear combinations over a finite field. In this paper we discuss selecting generations in a round-robin fashion: send one coded packet from each generation sequentially and
wrap around. 
As for the encoding scheme within the generations, we study three schemes: (1) the random linear combination approach as in \cite{generation_ITtrans}, (2) the random linear combination approach with a systematic phase, and (3) using an MDS (maximum distance separable) erasure code.

\subsubsection{Random Linear Combinations over GF($q$) (RL)}
Each block is represented as a row vector of symbols from finite field GF($q$), and the whole file is represented as a matrix of $N$ rows, one block each row. We abuse the notation a little here to use $G_j$ to denote the matrix representing generation $G_j$. To generate a coded block from a generation $G_j$ of $g$ blocks, choose a coding vector $c=[c_1,c_2,\dots,c_g]$ by choosing $g$ symbols independently and equiprobably from GF($q$). The resulting coded block is then $c\cdot G_j.$

\subsubsection{Random Linear Combinations Including a Systematic Phase (RLS)}
This is a variation of the RL scheme that includes a systematic phase at the beginning: send the original blocks from A to Z before starting to send
random linear combinations of the original file blocks.

\subsubsection{Maximum Distance Separable Codes (MDS)}
Over a finite field of small size, such as the common binary field, random linear combinations chosen in the way specified in the RL scheme inevitably introduces non-negligible linear dependency between the coded packets. For short lengths of data, we can use low-rate MDS codes instead. With an MDS($K$,$g$) code, $g$ packets are encoded into $K$ coded packets, and all the $g$ packets are recoverable as soon as any $g$ of the $K$ distinct coded packets have been collected. To extend transmission after the sender has exhausted all the $K$ coded packets, the sender repeats the coded packets in a round-robin fashion. The parity check code is a binary MDS code where $K=g+1.$ Reed-Solomon codes are another important class of MDS codes that operate on GF($2^l$) with $g<K<2^l.$ The increased complexity that comes with operations on a finite field of large size, however, can possibly undo the benefit brought by the MDS property, as we will later show in our experimental results.

\section{Statistics of the Delivery Packet Count}
\label{sec:theory}

In this section, we characterize $T$, the delivery packet count of coding within disjoint generations following the round-robin generation scheduling scheme. We assume that in each round, one coded packet is created from a generation that is selected from the $n$ generations in a wrap-around fashion. After the $t$th transmission, $m_t(=\lfloor t/n\rfloor)$ rounds have been completed. By that time, $(m_t+1)$ packets will have been sent from each of the first $[t-m_tn]$ generations, and $m_t$ packets from each of the rest $[(m_t+1)n-t]$ generations.

Since the generations are disjoint, each generation is decoded independently. Let $M_{g,\epsilon}$ be the number of coded packets
needed to be sent over a  link of
packet erasure rate $\epsilon$ from a generation of size $g$ so that the receiver can decode all file packets in the generation.
Let $p_{m,g,\epsilon}$ be the probability that $M_{g,\epsilon}\le m.$ Let $p_t$ be the probability that
$T\le t.$ 
Then, $p_t=p_{m_t+1,g,\epsilon}^{r_t}p_{m_t,g,\epsilon}^{n-r_t},$ where $m_t=\lfloor t/n\rfloor$ and $r_t=t-m_tn.$
Note that since $p_{m,g,\epsilon}$ is the cumulative probability function of $M_{g,\epsilon},$ $p_{m,g,\epsilon}$ is non-decreasing in $m,$ and hence $p_t$ is bounded as follows:
\begin{equation}
p_{m_t,g,\epsilon}^{n}\le p_t < p_{m_t+1,g,\epsilon}^{n}.
\end{equation}
Hence,
\begin{equation}\label{eq:T}E[T]=\sum_{t=0}^{\infty} (1-p_t)=\sum_{m=0}^{\infty}\Bigl(n-p_{m,g,\epsilon}\frac{p_{m+1,g,\epsilon}^{n-1}-p_{m,g,\epsilon}^{n-1}}{p_{m+1,g,\epsilon}-p_{m,g,\epsilon}}\Bigr)
\end{equation}
and
\begin{equation} n\sum_{m=1}^{\infty} (1-(p_{m,g,\epsilon})^{n})<E[T]\le n\sum_{m=0}^{\infty} (1-(p_{m,g,\epsilon})^{n}).\end{equation}


In the following, we characterize $p_{m,g,\epsilon}$ for different coding schemes within each generation.

\subsection{RL Scheme}
In this scheme, each coded packet is statistically the same; it is simply a random linear combination of the source packets.
To decode a generation of size $g$, a number $g$ of linearly independent coded packets must be received. When $m$ coded packets
have been transmitted
over the channel with erasure rate $\epsilon$, some $j\ge g$ have to be received, and among them $g$ have to be linearly independent.
Therefore, the probability $p_{m,g,\epsilon}^\text{RL}$ of successful decoding, given $m\ge g$ coded packets have been transmitted is given as follows:
\begin{claim}
\begin{equation}\label{eq:pm_rl}
p_{m,g,\epsilon}^\text{RL}= \sum_{j=g}^{m}{m\choose j}(1-\epsilon)^j\epsilon^{m-j}
\prod_{s=0}^{g-1}(1-q^{s-j})
\end{equation}
\end{claim}
The product in the equation is the probability that a $j\times g$ matrix with random entries chosen independently and
equiprobably from GF($q$) is of full column rank $g$. It is equal to the probability of having $g$ linearly independent
coded packets among $j$ coded packets. We can lower bound this product as follows (see \cite{Brent:2003}, Lemma~7):
\begin{equation}
\prod_{s=0}^{g-1}(1-q^{s-j}) \ge
\left\{
  \begin{array}{ll}
    0.288, & \hbox{if $q=2$ and $g=j$;} \\
    1-\frac{1}{q^{j-g}(q-1)}, & \hbox{otherwise.}
  \end{array}
\right.\label{eq:bound_qprod}
\end{equation}

When $q$ is large, we can further approximate $p_{m,g,\epsilon}^\text{RL}$ as follows:
\begin{align*}
p_{m,}^\text{RL}&{}_{g,\epsilon}
\gtrsim\sum_{j=g}^{m}{m\choose j}(1-\epsilon)^j\epsilon^{m-j}\\
&-\frac{1}{q-1}\sum_{j=g}^{m}{m\choose j}(1-\epsilon)^j(q\epsilon)^{m-j} q^{g-m}
\end{align*}
\subsection{RLS Scheme}
This scheme consists of two phases. In the first (systematic) phase, only uncoded packets are sent, and the second phase
is the same as the RL scheme described above. For each generation, first each of the original $g$ packets is transmitted once,
and random linear combinations of all the packets afterwards. Therefore, after $m$ transmissions, a generation can be decoded
if $l$ packets are received during the first phase of $g$ transmissions, and $g-l$ coded packets that are linearly independent
of the first $l$ are received in the second phase of $m-g$ transmissions. The probability $p_{m,g,\epsilon}^\text{RLS}$ of successful
decoding, given $m\ge g$ coded packets have been transmitted over the channel with the erasure rate $\epsilon$ is given as follows:

The probability of receiving $h$ linearly independent packets from $m$ transmissions is
\begin{claim}\label{thm:rls_pm}
\begin{align}\label{eq:pm_rls}
p_{m,g,\epsilon}^\text{RLS}=&(1-\epsilon)^g+\sum_{l=0}^{g-1} {g\choose l}(1-\epsilon)^l\epsilon^{g-l}p^{\text{RL}}_{m-g,g-l,\epsilon}\\
\gtrsim&\sum_{j=g}^m{m\choose j}(1-\epsilon)^j\epsilon^{m-j}+\frac{(1-\epsilon)^g}{q-1}(\frac{1-\epsilon}{q}+\epsilon)^{m-g}-\notag\\
&-\frac{1}{q-1}\sum_{j=g}^{m}{m\choose j}(1-\epsilon)^{j}\epsilon^{m-j}q^{g-j}\notag
\end{align}
\end{claim}
\begin{proof}
Please refer to the appendix.
\end{proof}

\subsection{MDS Scheme}
Suppose we use an MDS code which encodes $g$ symbols into $K$ symbols
s.t.\ the $g$ symbols can be entirely recovered as long as $g$ distinct symbols have been received.
We apply the code to generate $K$ encoded packets from $g$ original packets, and transmit the $K$ encoded packets in a round-robin fashion.

Let $u_m=\lfloor\frac{m}{K}\rfloor$ and $v_m=m-u_mK$ be the quotient and the remainder of the number of transmissions $m$ divided by the code block length $K$. Then, after $m$ transmissions, the first $v_m$ of the $K$ encoded packets have been transmitted $u_m+1$ times and the last $K-v_m$ of the $K$ encoded packets have been transmitted $u_m$ times. The probability that an encoded packet has been received is then $1-\epsilon^{u_m+1}$ for any packet among the first $v_m,$ and $1-\epsilon^{u_m}$ among the last $K-v_m.$
The probability $p^{\mathrm{MDS}}_{m,K,g,\epsilon}$ of successful decoding of all the $g$ packets, given $m$ encoded packets have been transmitted, is equal to the probability that at least $g$ of the $K$ encoded packets have been received, or at most $K-g$ encoded packets have never been received. Therefore, $p^{\mathrm{MDS}}_{m,K,g,\epsilon}$ can be computed by summing up the probability that $l$ of the first $v_m$ packets are absent and $j$ of the remaining $K-v_m$ packets are absent in the receiver collection for all integers $l$ and $j$ satisfying $0\le l+j\le K-g.$
\begin{claim}
\begin{align}\label{prob_MDS}
p^{\mathrm{MDS}}_{m,K,g,\epsilon}=&\sum_{l=0}^{K-g}{v_m\choose l}(\epsilon^{u_m+1})^l(1-\epsilon^{u_m+1})^{v_m-l}\\
&\cdot\sum_{j=0}^{K-g-l}{{K-v_m}\choose j}(\epsilon^{u_m})^j(1-\epsilon^{u_m})^{K-v_m-j}. \notag
\end{align}
where $u_m=\lfloor\frac{m}{K}\rfloor,$ $v_m=m-u_mK,$ and ${a\choose b}=0$ for $b>a.$
\end{claim}
When $m\le K,$ $u_m=0,$ $v_m=m,$ (\ref{prob_MDS}) becomes \[p^{\mathrm{MDS}}_{m,K,g,\epsilon}=\sum_{l=0}^{m-g}{m\choose l}\epsilon^l(1-\epsilon)^{m-l}=\sum_{j=g}^{m}{m\choose j}(1-\epsilon)^{j}\epsilon^{m-j}.\]

When $K=g,$ the code is the repetition code, and the right hand side of (\ref{prob_MDS}) becomes $(1-\epsilon^{u_m+1})^{v_m}(1-\epsilon^{u_m})^{K-v_m}.$

\section{Numerical and Experimental Results}
\label{sec:results}

To evaluate the performance of the schemes discussed in the previous section, we implemented them on an experimental platform consisting of a laptop computer and a smartphone. We measure the time and the energy consumption required for the receiver to recover the whole file. In this section, the experimental results are presented along with the theoretical predictions. 
\subsection{Experimental Setup}
The experimental setup consists of an HP Pavilion dv5-1120eg laptop computer as a transmitter and a Nokia N8 smartphone as a receiver. The specifications for the Nokia N8 are shown in Table~\ref{tab:specsn8}.
\begin{table}[hbt]
\caption{Specifications of the Nokia N8}
\label{tab:specsn8}
	\centering
		\begin{tabular} {|l|l|} \hline
			\bfseries Operating System & Symbian\textasciicircum 3 \\ \hline
			\bfseries CPU	& ARM11 @ 1 GHz \\ \hline
			\bfseries Memory & 256 MB SDRAM \\ \hline
			\bfseries Display	& 640 x 360 pixels, 3.5 inch \\ \hline
			\bfseries Battery & BL-4D (3.7 V, 1200 mAh Li-Ion) \\ \hline
		\end{tabular}
\end{table}

Both the laptop and the smartphone runs the same native C++ application (in sender and receiver mode, respectively) implemented using the Qt cross-platform application framework.
The laptop transmits a file at a nominal application-layer data rate of 1000KB/s via UDP and using IEEE 802.11b at a physical layer rate of 11~Mbps. A transmitted file consists of 512 random packets having 1400 data bytes each. These data packets are encoded following the three encoding schemes described in Section \ref{sec:model}.
The receiving cell phone tries to decode the original file without sending any feedback information to the sender except for a final completion indicator transmitted only when the file is fully decoded. The sender stops transmission once it has received this completion signal.


During the measurements the following information is recorded:
\begin{enumerate}
\item Number of packets sent before receiving the completion signal.
\item Number of packets received before sending the completion signal.
\item Time elapsed from the time when the first packet is received to the completion time.
\item Energy consumption by the receiver during the elapsed time. The test application uses the Control API of the Nokia Energy Profiler \cite{epnokia} to programmatically monitor (and record) the energy consumption of the mobile phone. The margin of error for these energy readings is 3\%.
\end{enumerate}

Each test was repeated 100 times for each generation size and encoding scheme pair. The following section presents the experimental results observed.

\subsection{Results}

\begin{figure}
\centering
	\includegraphics[width=\columnwidth]{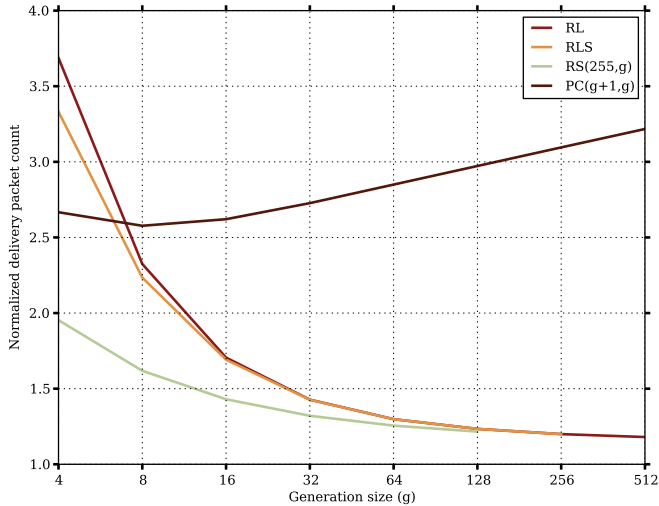}
\caption{Predicted expected delivery packet count (number of transmitted packets required to recover the entire file) versus generation size (assuming packet loss rate $\epsilon = 0.15$). RL: Random linear combinations. RLS: Random linear combinations with a systematic phase. RS(255,g): Reed-Solomon codes. PC($g+1$,$g$): A systematic code with a single coded packet as the bit-by-bit xor-sum of all file packets.}
\label{fig:dcount015}
\end{figure}

Figure~\ref{fig:dcount015} shows theoretical predictions for the normalized delivery packet count (i.e. how many packets are needed to successfully deliver one packet) under typical channel conditions in our experimental setup. The predictions are calculated from \eqref{eq:T} where $p_{m,g,\epsilon}$ is obtained from \eqref{eq:pm_rl}, \eqref{eq:pm_rls}, or \eqref{prob_MDS}. We observed that the packet loss rate ($\epsilon$) is around 15\% on an idle receiver when the sender is transmitting at a nominal rate of 1000KB/s. The RL and RLS schemes encode over the binary field, and the MDS schemes are represented by a Reed-Solomon (RS) code ($n=255, K=g$) and a simple Parity Check (PC) code ($n=g+1, K=g$) that has one parity symbol (all original symbols XORed together).
We observe that the overhead per packet drops as the generation size increases, and thus the probability of transmitting a packet for an already decoded generation decreases. This is not true for PC($g+1$,$g$) that can only cope with very low packet loss rates. The incorporation of a systematic phase in the random linear combination approach helps to reduce overhead for small generation sizes, but the gap quickly closes as the generation size increases. The Reed-Solomon code curve is near optimal since the code rates we use are much lower ($R<0.51$) than the packet loss rate, and with a high probability the transmission finishes before the sender runs out of the $255$ coded packets for each generation.

\begin{figure}
\centering
	\includegraphics[width=\columnwidth]{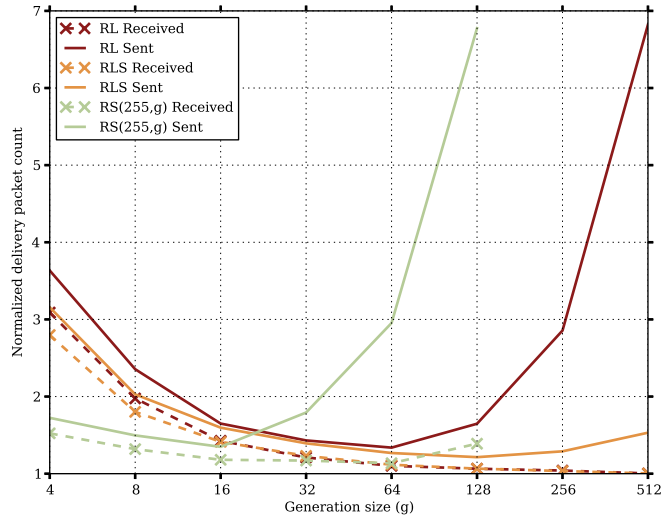}
\caption{Measured delivery packet count versus generation size}
\label{fig:exp1}
\end{figure}

Figure~\ref{fig:exp1} shows the average number of packets sent per successfully delivered packet as measured in our experiments. This was calculated using the total number of packets sent and received divided by the number of packets in the test file (i.e. $512$). For small generation sizes, we observe that increasing the generation size lowers the overhead per packet. These values are in accordance with the predictions in Figure~\ref{fig:dcount015}.
We would expect this trend to continue, since ideally we would use a single generation for the entire file. This would eliminate the possibility of transmitting packets that belong to an already decoded generation. However, this is not the optimal strategy in practice due to the increasing computational complexity. Figure~\ref{fig:exp1} shows that the overhead per sent packet increases significantly for the RS(255,$g$) scheme when $g > 16$, and for the RL scheme when $g > 64$. This indicates that the computational load on the receivers was too high, and they were unable to keep up with the transmission rate of the sender. The offset between the RS(255,$g$) and RL scheme may be explained by the larger field size used by the RS(255,$g$) scheme. Utilizing large fields (e.g. $q=2^8$) typically requires some form of memory based look-up table to perform multiplication and division, whereas all operations in the binary field ($q=2$) may be implemented using CPU instructions for binary XOR and AND operations. The RS implementation was based on a non-systematic Vandermonde matrix, other approaches such as utilizing binary Cauchy matrices~\cite{p:08:jer} should be considered to further increase the performance of this implementation.

The lower computational requirements associated with the systematic packets in the RLS scheme clearly benefit the overall system performance. It is however worth noting that the systematic phase assumes that the receivers did not  receive any packets previously. The systematic phase might lead to an additional overhead if the state of the receivers is initially unknown.
The curve of the RLS scheme only deviates from the predicted values for very high generation sizes, $256$ and $512$.

Figure~\ref{fig:exp2} shows that when we plug the average (application-layer) packet loss rate observed from the experiments (the loss rates are higher for larger generation sizes) into Claims 1-3, the theoretical predictions still match experimental data. This confirms the validity of our characterization.

\begin{figure}
\centering
	\includegraphics[width=\columnwidth]{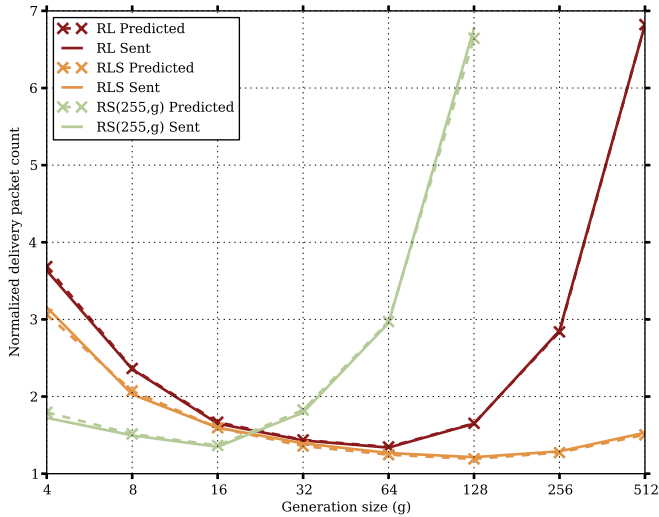}
\caption{Measured delivery packet count compared to predictions calculated for measured packet loss rates}
\label{fig:exp2}
\end{figure}

\begin{figure}
\centering
	\includegraphics[width=\columnwidth]{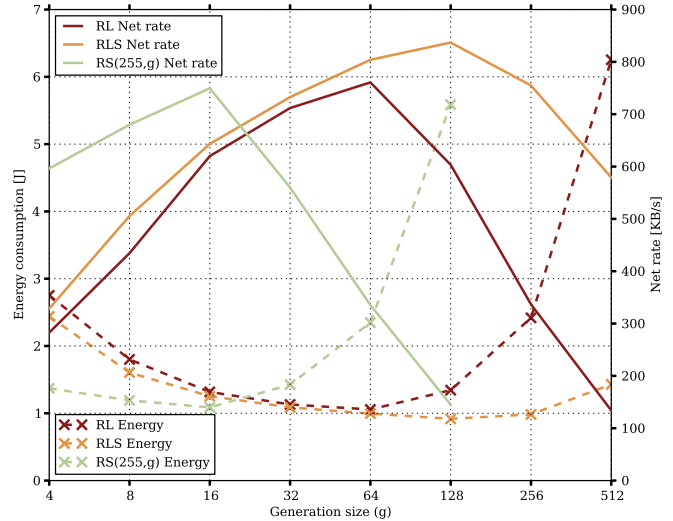}
\caption{Net data rate ($=\frac{\textnormal{file size}}{\textnormal{delivery time}}$) and energy consumption versus generation size}
\label{fig:exp3}
\end{figure}

The energy consumption of the communication system is especially important on battery-driven mobile devices. In Figure~\ref{fig:exp3}, we show the average energy consumption in Joules per file download. This is compared to the net throughput observed throughout the test. Due to the dominant impact of the wireless radio on the power consumption, we observe a significant connection between these two measured quantities. As the power consumption of the wireless radio remains relatively stable, when not in power-save or sleep mode, the energy consumption largely depends on the time needed to complete a test, and thus it is inversely proportional to the net rate. In order to minimize the energy consumption of the protocol, we need to maximize the net rate.

When comparing the RS(255,$g$) to the binary-field RL and RLS codes, we observe that using a larger field size yields a better code performance at lower generation sizes. On the other hand, it is unable to sustain the low overhead as the generation size and thereby the computational complexity increases.

Although these results and the specific optimal values are certainly device- and system-dependent, we expect that  other devices would exhibit similar tendencies, but the actual values would be shifted depending on the capabilities of the given platform. Faster devices might be able to support higher generation sizes and higher data rates.



\section{Conclusion and Future Work}
\label{sec:concl}

In this paper, we considered three application-layer coding schemes for streaming over lossy links: random linear coding (RL), systematic random linear coding (RLS), and structured coding (MDS). We characterized the exact distribution and the expected value of the delivery packet count of coding within disjoint generations following the round-robin generation scheduling scheme, taking into account the effect of field size and generation size. Our characterization matches experimental results.

The three coding schemes were implemented on a laptop computer and a Nokia N8 smartphone using the Qt cross-platform application framework. We presented measurement results collected during numerous experiments with various settings. Results show that the computational complexity has a significant impact on the performance of these schemes. The RLS scheme is the least computationally intensive, thereby it is able to achieve the highest net data rate and the lowest energy consumption.

In the future, we plan to implement other codes such as LT codes, Raptor codes, and systematic Reed-Solomon codes on the same testbed in order to compare their performance to the coding schemes discussed in this paper.
The cost of random memory access and finite field operations is non-negligible on a terminal with constrained capacity.
A model should be devised that can account for these factors to give predictions on other platforms with different capabilities and constraints.

\section*{Acknowledgments}
This work was partially financed by the CONE project (Grant No. 09-066549/FTP) granted by the Danish Ministry of Science, Technology and Innovation as well as by the collaboration with Renesas Mobile throughout the NOCE project.

\appendix
\section*{Proof of Claim \ref{thm:rls_pm}}
\begin{align}
&\sum_{l=0}^{g-1} {g\choose l}(1-\epsilon)^l\epsilon^{g-l}p^{\text{RL}}_{m-g,g-l,\epsilon} \notag\\
&=\sum_{l=0}^{g-1} {g\choose l}(1-\epsilon)^l\epsilon^{g-l}\cdot\\[-1.5mm]
&\qquad\qquad\cdot\sum_{j=g-l}^{m-g}{{m-g}\choose j}(1-\epsilon)^j\epsilon^{m-g-j}
\prod_{s=0}^{g-l-1}(1-q^{s-j}) \notag\\
&=\sum_{l=0}^{g-1}\!\! {g\choose l}\sum_{j=g-l}^{m-g}{{m-g}\choose j}(1-\epsilon)^{j+l}\epsilon^{m-l-j}
\!\!\prod_{s=0}^{g-l-1}(1-q^{s-j})\notag\\
&=\sum_{l=0}^{g-1} {g\choose l}\sum_{j=g}^{m-g+l}{{m-g}\choose {j-l}}(1-\epsilon)^{j}\epsilon^{m-j}
\!\!\prod_{s=0}^{g-l-1}(1-q^{s-j+l})\notag\\
&=\sum_{j=g}^{m-1}\sum_{l=0}^{g-1} {g\choose l}{{m-g}\choose {j-l}}(1-\epsilon)^{j}\epsilon^{m-j}
\!\!\prod_{s=0}^{g-l-1}(1-q^{s-j+l})\notag\\
&\gtrsim\sum_{j=g}^{m-1} \sum_{l=0}^{g-1} {g\choose l}{{m-g}\choose {j-l}} (1-\epsilon)^{j}\epsilon^{m-j}
(1-\frac{1}{q-1}q^{g-j})\tag{*}\label{eq:gtrsim}
\end{align}
\eqref{eq:gtrsim} follows from (\ref{eq:bound_qprod}). Applying Vandermonde's identity $\sum_{l=0}^g {g\choose l}{{m-g}\choose {j-l}}={m\choose j}$ to \eqref{eq:gtrsim}, we have
\begin{align*}
 &(*)\\
 &=\sum_{j=g}^{m}\Big({m\choose j}-{{m-g}\choose{j-g}}\Big) (1-\epsilon)^{j}\epsilon^{m-j}
(1-\frac{1}{q-1}q^{g-j})\\
&=\sum_{j=g}^{m}{m\choose j}(1-\epsilon)^j\epsilon^{m-j}-\sum_{j=g}^{m}{{m-g}\choose{j-g}}(1-\epsilon)^j\epsilon^{m-j}\\
&-\frac{1}{q-1}\sum_{j=g}^{m}{m\choose j}(1-\epsilon)^{j}\epsilon^{m-j}q^{g-j}\\
&+\frac{1}{q-1}\sum_{j=g}^{m}{{m-g}\choose{j-g}}(1-\epsilon)^{j}\epsilon^{m-j}q^{g-j}\\
&=\sum_{j=g}^{m}{m\choose j}(1-\epsilon)^j\epsilon^{m-j}-(1-\epsilon)^g\sum_{j=0}^{m-g}{{m-g}\choose{j}}(1-\epsilon)^j\epsilon^{m-g-j}\\
&-\frac{1}{q-1}\sum_{j=g}^{m}{m\choose j}(1-\epsilon)^{j}\epsilon^{m-j}q^{g-j}\\
&+\frac{(1-\epsilon)^g}{q-1}\sum_{j=0}^{m-g}{{m-g}\choose{j}}(\frac{1-\epsilon}{q})^{j}\epsilon^{m-g-j}\\
&=\sum_{j=g}^{m}{m\choose j}(1-\epsilon)^j\epsilon^{m-j}-\frac{1}{q-1}\sum_{j=g}^{m}{m\choose j}(1-\epsilon)^{j}\epsilon^{m-j}q^{g-j}\\
&-(1-\epsilon)^g+\frac{(1-\epsilon)^g}{q-1}(\frac{1-\epsilon}{q}+\epsilon)^{m-g}.
\end{align*}

\bibliographystyle{IEEEtran}
\bibliography{roundrobin}

\end{document}